\documentclass[aps,prl,reprint,superscriptaddress]{revtex4-1}

\usepackage{amsmath,amssymb,amsthm}

\usepackage[thinlines,thiklines]{easybmat}

\usepackage{verbatim,epsfig}

\usepackage{color}

\usepackage{stackengine}

\usepackage{framed}
{\begin{leftbar}\underline{\textbf{Internal Notes}}:\begin{quotation}}%
{\end{quotation}\end{leftbar}}

\usepackage{mdwlist}

\newtheorem{theorem}{Theorem}
\newtheorem{proposition}{Proposition}
\newtheorem{lemma}{Lemma}

\newtheorem{corollary}{Corollary}
\theoremstyle{definition}

\newcommand{\ket}[1]{|#1\rangle}

\newcommand{\op}[2]{|#1\rangle \langle #2|}

\DeclareMathOperator{\tr}{Tr}

\newcommand{\mc}[1]{\mathcal{#1}}

\begin{document}

\title{A Classical Analog to Entanglement Reversibility}

\author{Eric Chitambar}
\affiliation{Department of Physics and Astronomy, Southern Illinois University, Carbondale, Illinois 62901, USA}
\author{Ben Fortescue}
\affiliation{Department of Physics and Astronomy, Southern Illinois University, Carbondale, Illinois 62901, USA}
\author{Min-Hsiu Hsieh}
\affiliation{Centre for Quantum Computation \& Intelligent Systems (QCIS), Faculty of Engineering and Information Technology (FEIT), University of Technology Sydney (UTS), NSW 2007, Australia}

\date{\today}

\begin{abstract}
In this letter we introduce the problem of secrecy reversibility.  This asks when two honest parties can distill secret bits from some tripartite distribution $p_{XYZ}$ and transform secret bits back into $p_{XYZ}$ at equal rates using local operation and public communication (LOPC).  This is the classical analog to the well-studied problem of reversibly concentrating and diluting entanglement in a quantum state.  We identify the structure of distributions possessing reversible secrecy when one of the honest parties holds a binary distribution, and it is possible that all reversible distributions have this form.  These distributions are more general than what is obtained by simply constructing a classical analog to the family of quantum states known to have reversible entanglement.  An indispensable tool used in our analysis is a conditional form of the G\'{a}cs-K\"{o}rner common information. 

\end{abstract}

\maketitle

Resource theories offer a powerful framework for studying what physical processes are possible under a certain class of constraints.  For instance, when studying the manipulation of quantum systems, entanglement is identified as a precious resource that cannot be freely generated under local quantum operations and classical communication (LOCC).  Inspired by the conceptual successes of entanglement theory, researchers have recently begun applying a resource-theoretic perspective toward the notion of secrecy in classical information theory \cite{Collins-2002a, Horodecki-2005c}.  For two-party secrecy, one considers tripartite distributions $p_{XYZ}$: Alice ($X$) and Bob ($Y$) share correlations about which, undesirably, Eve ($Z$) has side information.  The distributions are manipulated using local operations and public communication (LOPC), which is the classical analog of LOCC.  Just as the ebit $\ket{\Phi}=\sqrt{1/2}(\ket{00}+\ket{11})$ represents a fundamental unit of entanglement, a secret bit $\Phi_{XY}\cdot q_Z$ represents a fundamental unit of secrecy.  Here, $\Phi_{XY}(i,j)=(1/2)\delta_{ij}$ is a perfectly correlated bit while $q_Z$ is an arbitrary and uncorrelated distribution.

Quantum entanglement and classical secrecy share many striking similarities \cite{Collins-2002a, Gisin-2002b, Acin-2003b, Horodecki-2005c, Acin-2005b, Christandl-2007a, Oppenheim-2008a, Bae-2009a, Ozols-2014a}.  One important similarity lies in the tasks of resource distillation and resource cost.  For a bipartite quantum state $\rho_{AB}$, its \textit{distillable entanglement} $E_D(\rho_{AB})$ quantifies, roughly speaking, the amount of ebits that can be distilled from $\rho_{AB}$ using LOCC \cite{Rains-1999a} (in the many-copy sense), while its \textit{entanglement cost} $E_C(\rho_{AB})$ quantifies the amount of ebits required to generate $\rho_{AB}$ using LOCC \cite{Hayden-2001a}.  For a distribution $p_{XYZ}$, its ``secrecy content'' can analogously be quantified in terms of its distillable key $K_D(p_{XYZ})$ \cite{Ahlswede-1993a, Maurer-1993a} and its key cost $K_C(p_{XYZ})$ \cite{Renner-2003a}.  Here, the distillation goal is to obtain secret bits $\Phi_{XY}$ from $p_{XYZ}$, while the formation goal is simulate $p_{XYZ}$ using $\Phi_{XY}$ and public communication.  Compared to entanglement theory, much less is known about the relationship between $K_D$ and  $K_C$, except for the expected hierarchy $K_C\geq K_D$ \cite{Renner-2003a}.  

With the inequality $K_C\geq K_D$, classical secrecy can be given a thermodynamic interpretation similar to entanglement \cite{Popescu-1997a, Horodecki-2002c}.  By the second law of thermodynamics, a heat engine cannot do more work when transferring heat from one temperature bath to a lower one than the work required to perform the reverse refrigeration process.  Likewise, $K_C(p_{XYZ})\geq K_D(p_{XYZ})$ means that an LOPC protocol is not able to distill more secret bits from $p_{XYZ}$ than the secret bits needed to perform the reverse formation process.  Any distribution for which this inequality is tight can thus be regarded as the secrecy analog of a reversible heat engine. The \textit{secrecy reversibility problem} asks what distributions satisfy $K_C(p_{XYZ})=K_D(p_{XYZ})$.   

To begin tackling this problem, it is instructive to first consider the quantum scenario.  It is well-known that all bipartite quantum pure states demonstrate entanglement reversibility: any pure state can be concentrated into an EPR state $\ket{\Phi}$ and diluted back to the original state at equal rates \cite{Bennett-1996b}.  Thus, a natural starting place to find reversible secrecy is with a classical analog to quantum pure states.  Collins and Popescu have investigated \cite{Collins-2002a} one such analog based on an embedding of $p(x,y,z)$ into a tripartite quantum state given by
\begin{equation}
\label{Eq:qqqEmbedding}
\ket{\Psi}_{ABE}=\sum_{x,y,z}\sqrt{p(x,y,z)}\ket{xyz}.
\end{equation}  
If Alice and Bob's reduced state in $\ket{\Psi}$ is pure, then $\ket{\Psi}$ can always be expressed as $\ket{\Psi}=\sum_{j,z}\sqrt{p(j)q(z)}\ket{\alpha_j\beta_j}\ket{z}$, where $\ket{\alpha_j}$ and $\ket{\beta_j}$ are Schmidt basis vectors.
With this motivation, Collins and Popescu have proposed distributions of the form $p(x,y,z)=\delta_{xy}p(x)q(z)$ to be the classical analog to quantum pure states (another type of analog has also been proposed in the literature  \footnote{In Ref.~\cite{Oppenheim-2008a}, the authors introduce another class of distributions, called \textit{bi-disjoint}, which they propose as a different analog to quantum pure states.  In terms of the notation used here, a distribution is bi-disjoint iff $I(XY:Z|J_{(XY)Z})=0$, where $J_{(XY)Z}$ is the common information between Alice-Bob (jointly) and Eve.  It is shown in \cite{Oppenheim-2008a} that bi-disjoint distributions behave like quantum pure states for the task of state merging.  However, in general they fail to possess reversible secrecy, nor do they behave like quantum pure states for single-copy state transformations.  It therefore seems that classical analogies to quantum pure states can only be drawn \textit{with respect to specific information-theoretic tasks}, a conclusion already implicitly acknowledged in \cite{Oppenheim-2008a}.  For the task of resource reversibility, SBI distributions are the more appropriate analog to quantum pure states.}).  Actually, we can generalize the Collins-Popescu class of distributions to include distributions of the form
\begin{equation}
\label{Eq:SBI}
p(x,y,z)=\sum_jp(x|j)p(y|j)p(j)q(z),
\end{equation}
where $p(x|j)p(x|j')=p(y|j)p(y|j')=0$ if $j\not=j'$.  A quantum embedding of any such distribution \textit{\`{a} la} Eq. \eqref{Eq:qqqEmbedding} recovers a pure state for Alice and Bob with Schmidt basis vectors $\ket{\alpha_j}=\sum_x\sqrt{p(x|j)}\ket{x}$ and $\ket{\beta_j}=\sum_y\sqrt{p(y|j)}\ket{y}$.  We refer to any distribution having the the form of Eq. \eqref{Eq:SBI} as \textit{secret block independent} (SBI), and they may also be considered as a type of ``classical pure state.''  Like the Collins-Popescu distributions, the theory of single-copy state transformations can be constructed for SBI states analogous to pure quantum states \cite{Collins-2002a}.  Furthermore, just as pure quantum states possess reversible entanglement, SBI distributions possess reversible secrecy, as will be shown below (also see footnote 12 in \cite{Oppenheim-2008a}).   

However, it turns out that a much richer class of distributions beyond SBI also demonstrate secrecy reversibility.  To begin studying this class, we first recall a well-known upper bound on $K_D(p_{XYZ})$ referred to as the \textit{intrinsic information} of  $p_{XYZ}$ \cite{Maurer-1999a}.  This quantity is given by 
\begin{equation}
\label{Eq:IntrinsicInfo}
I(X:Y\downarrow Z):=\min I(X:Y|\overline{Z}),
\end{equation}
where the minimization is taken over over all auxiliary variables $\overline{Z}$ such that $XY-Z-\overline{Z}$ forms a Markov chain \footnote{Recall, a triple of random variables $ABC$ ranging over $\mc{A}\times\mc{B}\times\mc{C}$ form a Markov chain $A-B-C$ if $p(a|bc)=p(a|b)$ for all $(a,b,c)\in \mc{A}\times\mc{B}\times\mc{C}$ such that $p(b,c)>0$.  In entropic terms, this is equivalent to the vanishing of the conditional mutual information: $I(A,C|B)=0$.}.  Using the definition of key cost, Renner and Wolf were able to prove that $K_C(p_{XYZ})\geq I(X:Y\downarrow Z)$ \cite{Renner-2003a}, and thus 
\begin{equation}
\label{Eq:Dist-Intrinsic-Cost}
K_D(p_{XYZ})\leq I(X:Y\downarrow Z)\leq K_C(p_{XYZ}).
\end{equation}
Consequently, we can split the secrecy reversibility problem into two separate questions: (1) when does $K_C(p_{XYZ})=I(X:Y\downarrow Z)$, and (2) when does $K_D(p_{XYZ})=I(X:Y\downarrow Z)$?  We answer the first question below and reference certain results from Ref. \cite{Chitambar-2014c} where we have recently studied the second question.  However, before doing so, we introduce a variety of distribution classes based on the notion of a conditional common function since these classes will play a central role in our analysis of reversible secrecy.


\medskip

\noindent\textit{Common Functions and UBI-PD$\downarrow$ Distributions.}
For distribution $p_{XY}$, a \textit{maximal common function} is a variable $J_{XY}$ such that
\begin{equation}
\label{Eq:GK-Common-Info1}
H(J_{XY})=\max_K \{H(K):0=H(K|X)=H(K|Y)\}.
\end{equation}
The value $H(J_{XY})$ has been identified by G\'{a}cs and K\"{o}rner as the \textit{common information} between $X$ and $Y$ \cite{Gacs-1973a}.  It can be shown that for every $p_{XY}$, the variable $J_{XY}$ is unique up to a relabeling of its range (see Supplemental Material).  Note that an SBI distribution can be equivalently characterized by the entropic condition $I(X:Y|J_{XY})=0$, and $H(J_{XY})=I(X:Y)$ for these distributions \cite{Gacs-1973a}.

For a tripartite distribution $p_{XYZ}$, we will denote a maximal common function of the conditional distribution $p_{XY|Z=z}$ by $J_{XY|Z=z}$.  Then, a \textit{maximal conditional common function} $J_{XY|Z}$ is just a collection of maximal common functions $\{J_{XY|Z=z}:p(z)>0\}$.  Again, the variable $J_{XY|Z}$ is unique up to relabeling.  We say that a distribution $p_{XYZ}$ is \textit{block independent} (BI) if $I(X:Y|J_{XY|Z}Z)=0$;
equivalently, if the distribution decomposes as
\begin{equation}
\label{Eq:BIform}
p(x,y,z)=\sum_{z\in\mc{Z}}\sum_{J_{XY|Z=z}=j}p(x|z,j)p(y|z,j)p(j,z),
\end{equation}
where $p(x|z,j)p(x|z,j')=0$ and $p(y|z,j)p(y|z,j')=0$ for $j\not=j'$.  Obviously SBI distributions are simply BI with an uncorrelated Eve.  A distribution is said to be \textit{uniform block independent} (UBI) if it is block independent, and there exist local coarse-graining maps $K_X(X)$ and $K_Y(Y)$ such that $Pr[J_{XY|Z}=K_X=K_Y]=1$ for some maximal common function $J_{XY|Z}$.  In other words, Alice and Bob can determine the value for $J_{XY|Z}$ simply by consulting their local variable.  With many copies of a UBI distribution, secret key can be distilled via privacy amplification at an optimal rate $H(J_{XY|Z}|Z)=I(X:Y|Z)$ \cite{Ahlswede-1993a, Bennett-1995a}.  

However, in general $J_{XY|Z}$ will be unknown to Alice and Bob unless they engage in public communication.  A public communication protocol is a sequence of public messages $M=(M_1,M_2,\cdots,M_r)$ such that $M_k$ is a function of both $M_{k-1}\cdots M_1$ and $X$ (resp. $Y$) when $k$ is odd (resp. even).  At the end of these exchanges, the new object of interest becomes $J_{XY|ZM}$, which is a maximal conditional common function for the distribution $p_{(XM)(YM)(ZM)}$.  It can easily be proven that when $p_{XYZ}$ is BI, so is $p_{(MX)(MY)(ZM)}$, and furthermore
\begin{align}
\label{Eq:CCI-Messages-Equality}
I(X:Y|ZM)=I(X:Y|Z)-I(M:J_{XY|Z}|Z)
\end{align}
(see Supplemental Material).  This equation formalizes the intuitive idea that messages $M$ will decrease Alice and Bob's average conditional common information unless, from Eve's perspective, the messages are independent of $J_{XY|Z}$.  

With this motivation, we say $p_{XYZ}$ is \textit{uniform block independent under public discussion} (UBI-PD) if it is BI and there is a public communication protocol generating messages $M$ such that $p_{(MX)(MY)(ZM)}$ is UBI and $I(M:J_{XY|Z}|Z)=0$.  Thus, UBI-PD distributions have a distillation rate of $H(J_{XY|ZM}|ZM)$, which by Eq. \eqref{Eq:CCI-Messages-Equality} is equal to $H(J_{XY|Z}|Z)=I(X:Y|Z)$.  We say a distribution belongs to the class UBI-PD$\downarrow$ if there exists a channel $\overline{Z}|Z$ such that $p_{XY|\overline{Z}}$ is UBI with the required public communication $M$ also satisfying $I(Z:J_{XY|\overline{Z}}|M\overline{Z})=0$.  This latter condition assures that $K_D(p_{XYZ})=H(J_{XY|\overline{Z}}|\overline{Z})=I(X:Y\downarrow Z)$.  Indeed, for every UBI-PB$\downarrow$ distribution, $J_{XY|\overline{Z}}$ becomes shared randomness under communication $M$.  Thus, an achievable key rate is 
\[H(J_{XY|\overline{Z}}|ZM)=H(J_{XY|\overline{Z}}|\overline{Z}M)=H(J_{XY|\overline{Z}}|\overline{Z}),\] where the first equality follows from $I(Z:J_{XY|\overline{Z}}|M\overline{Z})=0$ and the second from $I(M:J_{XY|\overline{Z}}|\overline{Z})=0$.  Fig. \ref{Fig:UBIPDdownarrow} depicts a UBI-PD$\downarrow$ distribution.  We encourage the reader to visit the Supplemental Material for a comparative picture of the various distribution classes identified here.  

\begin{figure}
\centering
     \includegraphics[scale=1.1]{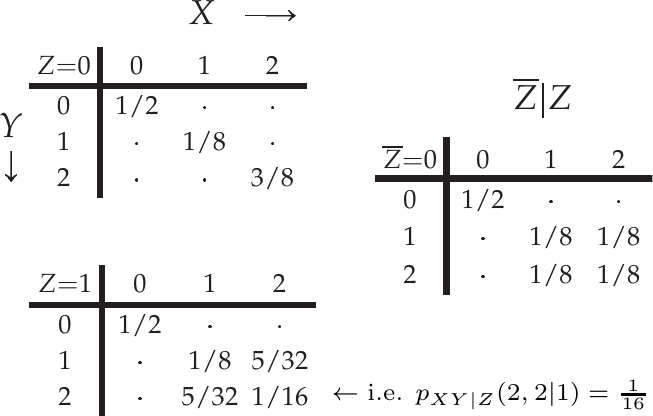}
     \caption{A UBI-PD$\downarrow$ distribution where $p_Z(0)=1/5$, $p_Z(1)=4/5$ and $\overline{Z}$ is a full coarse-graining of $Z$.  In this simplified example, no communication is needed for Alice and Bob to both generate $J_{XY|\overline{Z}}$.  Note that $p_{XYZ}$ itself is not BI.}
		\label{Fig:UBIPDdownarrow}
      \end{figure}

\medskip
\noindent\textit{When does $K_C(p_{XYZ})=I(X:Y\downarrow Z)$?}
This question can be answered using the formula for key cost as computed by Winter, a significant result on its own since no single-letter expression is known for $K_D(p_{XYZ})$.
\begin{lemma}[Winter \cite{Winter-2005a}]
\label{Lem:WinterKC}
For a distribution $p_{XYZ}$,
\begin{equation}
\label{Eq:WinterKC}
K_C(p_{XYZ})=\min I(XY:W|\overline{Z}),
\end{equation}
where the minimization is over all auxiliary variables $W$ and $\overline{Z}$ which satisfy $XY-Z-\overline{Z}$ and $X-W\overline{Z}-Y$.  
\end{lemma}
\noindent Our task will now be to reproduce Renner and Wolf's result that $K_C(p_{XYZ})\geq I(X:Y\downarrow Z)$ directly from Winter's formula \eqref{Eq:WinterKC}.  In doing so, we will obtain a structure condition for when $K_C(p_{XYZ})=I(X:Y\downarrow Z)$.

\begin{lemma}
\label{Lem:Structure1}
For the distribution $p_{XYZ}$, $K_C(p_{XYZ})\geq I(X:Y\downarrow Z)$.  Equality is obtained iff $p_{XY\overline{Z}}$ is BI, where $\overline{Z}|Z$ is the minimizer in $I(X:Y\downarrow Z)$.  When equality holds, $K_C(p_{XYZ})=I(X:Y\downarrow Z)=H(J_{XY|\overline{Z}}|\overline{Z})$.  
\end{lemma}

\begin{proof}
Let $XYZW\overline{Z}$ satisfy the minimization in Eq.~\eqref{Eq:WinterKC}.  Then we have the following chain of inequalities:
\begin{align}
\label{Eq:Cost-Ineq-chain}
K_C(p_{XYZ})&=I(XY:W|\overline{Z})\geq I(X:W|\overline{Z})\notag\\
&\qquad\geq I(X:Y|\overline{Z})\geq I(X:Y\downarrow Z).
\end{align}
The first inequality follows from the fact that $I(Y:W|X\overline{Z})\geq 0$ with equality obtained iff $W-X\overline{Z}-Y$; the second inequality is the data-processing inequality applied to $X-W\overline{Z}-Y$ with equality obtained iff $X-Y\overline{Z}-W$; and the third inequality follows from the definition of intrinsic information.

For the equality conditions, consider when $X-Y\overline{Z}-W$ and $Y-X\overline{Z}-W$.  This so-called ``conditional double Markov chain'' can only be satisfied if $I(XY:W|J_{XY|\overline{Z}}\overline{Z})=0$ (see Supplemental Material).  Using this we upper bound the key cost by
\begin{align}
\label{Eq:KeyCostCond1}
I(XY:W|\overline{Z})&=I(XYJ_{XY|\overline{Z}}:W|\overline{Z})\notag\\
&=I(J_{XY|\overline{Z}}:W|\overline{Z})\leq H(J_{XY|\overline{Z}}|\overline{Z}).
\end{align}
Since $J_{XY|Z}$ is both a function of $X$ and $Y$ given $Z$, it is easy to show $H(J_{XY|\overline{Z}}|\overline{Z})\leq I(X:Y|\overline{Z})$, with equality iff $I(X:Y|\overline{Z}J_{XY|\overline{Z}})=0$.  Hence demanding that $I(XY:W|\overline{Z})=I(X:Y|\overline{Z})$ gives the necessary conditions $H(J_{XY|\overline{Z}}|W\overline{Z})=0$ and $I(X;Y|J_{XY|\overline{Z}}\overline{Z})=0$.  

Conversely, if $p_{XY\overline{Z}}$ is block independent and $I(X:Y\downarrow Z)=I(X:Y|\overline{Z})$, then choose $\overline{Z}$ and $W=J_{XY|\overline{Z}}$ in the minimization of Eq.~\eqref{Eq:WinterKC} to obtain $K_C(p_{XYZ})=I(X:Y\downarrow Z)$.
\end{proof}


\noindent\textit{A Class of Reversible Distributions.}  We have seen that $K_D(p_{XYZ})=I(X:Y\downarrow Z)$ for UBI-PD$\downarrow$ distributions.  Since these distributions admit a channel $\overline{Z}|Z$ with $p_{XY\overline{Z}}$ being BI, Lemma \ref{Lem:Structure1} gives that $K_D(p_{XYZ})=K_C(p_{XYZ})$ for every UBI-PD$\downarrow$ distribution.  We have thus identified a family of distributions possessing reversible secrecy, and we conjecture that this family completely characterizes secrecy reversibility in the classical setting.  The conjecture obviously holds true for any distribution with $0=K_C(p_{XYZ})=K_D(p_{XYZ})$ since $K_C(p_{XYZ})=0$ implies $I(X:Y\downarrow Z)=0$ by Lemma \ref{Lem:Structure1}, and any distribution satisfying the latter condition is UBI-PB$\downarrow$ by definition.  The conjecture can also be shown as true for distributions satisfying $\min\{|\mc{X}|,|\mc{Y}|\}=2$.  
\begin{theorem}
\label{Thm:Reversible-Structure-2x2}
If $\min\{|\mc{X}|,|\mc{Y}|\}=2$, then $K_C(p_{XYZ})=K_D(p_{XYZ})$ iff $p_{XYZ}$ is UBI-PD$\downarrow$.
\end{theorem}
\begin{proof}
Here we prove the theorem for when $|\mc{X}|=|\mc{Y}|=2$, and the more general case is handled in the Supplemental Material.  Crucial to our argument is a necessary structural condition recently proven for distributions satisfying $K_D(p_{XYZ})=I(X:Y|Z)$ \cite{Chitambar-2014c}.  
\begin{proposition}[\cite{Chitambar-2014c}]
\label{Prop:DistOpt1-2x2}
When $|X|=|Y|=2$ and there exists a pair $(x,y)$ such that $p(x,y|z_1)p(x|z_0)p(y|z_0)>0$ but $p(x,y|z_0)=0$ for some $z_0,z_1\in\mc{Z}$, then $K_D(p_{XYZ})<I(X:Y|Z)$.
\end{proposition}
\noindent Continuing with the proof of Theorem \ref{Thm:Reversible-Structure-2x2} in the $2\times 2$ case, from the previous discussion it suffices to prove necessity when $K_C(p_{XYZ})=K_D(p_{XYZ})>0$.  Then by Lemma \ref{Lem:Structure1}, for some $\overline{Z}|Z$, $p_{XY\overline{Z}}$ must be block independent and $K_D(p_{XYZ})=I(X:Y|\overline{Z})$.  However, since $K_D(p_{XYZ})\leq K_D(p_{XY\overline{Z}})\leq I(X:Y|\overline{Z})$, we see that $ K_D(p_{XY\overline{Z}})= I(X:Y|\overline{Z})$.  Then from Proposition \ref{Prop:DistOpt1-2x2}, the structure of BI distributions, and the fact that $H(J_{XY|\overline{Z}}|\overline{Z}=z)>0$ for some $z$, we have that $H(X|Y)=H(Y|X)=0$; i.e. $p_{XY\overline{Z}}$ is UBI and, up to a relabeling, has the form $p(x,y,z)=\delta_{xy}[xq(z)+(1-x)(1-q(z))]$.  Since $\overline{Z}$ is obtained by processing $Z$, $p_{XY\overline{Z}}$ can have this correlated form only if $p_{XYZ}$ likewise does.  Thus, $p_{XYZ}$ is UBI.
\end{proof}

\noindent\textit{Reversible Distributions Embedded in Quantum States.}  We now consider embedding reversible distributions into quantum states as in Eq. \eqref{Eq:qqqEmbedding}.  In particular, we focus on distributions with $|\mc{X}|=|\mc{Y}|=2$ so that the corresponding $\rho_{AB}:=\tr_E\op{\Psi}{\Psi}_{ABE}$ is a two-qubit state.  We can make a comparison between the secret key of the underlying distribution and the entanglement of the embedded quantum state using an analytic formula for the entanglement of formation $E_F$\cite{Wootters-1998a}.  The following relatively straightforward calculation is carried out in the Supplemental Material. 
\begin{theorem}
\label{Thm:Gap2x2}
For reversible $p_{XYZ}$ with $|\mc{X}|=|\mc{Y}|=2$ and $K_D(p_{XYZ})>0$:
\begin{align}
K_D(p_{XYZ})&=\sum_{z\in\mc{Z}}p(z)\mathsf{E}\left(2\sqrt{p(0|z)p(1|z)}\right)\notag\\
E_F(\rho_{AB})&=\mathsf{E}\left(2\sum_{z\in\mc{Z}}p(z)\sqrt{p(0|z)p(1|z)}\right),
\end{align}
where $\mathsf{E}(x):=h(\tfrac{1}{2}[1-\sqrt{1-x^2}])$ is strictly convex in $x$ for $h(x):=-x\log x-(1-x)\log(1-x)$.  The equality $K_D(p_{XYZ})=E_F(\rho_{AB})$ holds iff $H(X|Z=z)$ is constant for all $z\in\mc{Z}$. 
\end{theorem}

It is natural to wonder whether a quantum state with an embedded reversible distribution will likewise possess reversible entanglement.  However, one can already see in two qubits that this will not be true in general.  Every two-qubit embedded $\rho_{AB}$ with $K_D(p_{XYZ})>0$ will take the form $\rho_{AB}=\sum_z\sum_{j,j'=0}^1p(z)\sqrt{p(j|z)p(j'|z)}\op{jj}{j'j'}$.  This is a so-called maximally-correlated state for which entanglement reversibility is known to be lacking whenever $\rho_{AB}$ is not pure \cite{Cornelio-2011a, Vollbrecht-2004a}.  In fact, $E_F(\rho_{AB})$ is additive for the states of Theorem \ref{Thm:Gap2x2} \cite{Vidal-2002c}.  Thus,
\begin{corollary}
\label{Cor:ECEDgap}
When $|\mc{X}|=|\mc{Y}|=2$, any distribution with nonzero reversible secrecy will have nonzero reversible entanglement when embedded in a quantum state iff the embedded state is pure.  
\end{corollary}

\noindent\textit{Returning to Reversible Entanglement.}  We motivated our investigation into reversible secrecy by considering reversible entanglement in quantum pure states and asking for a classical analog.  This led  to the proposal of SBI distributions as being a type of ``classical pure state.''  Beyond pure states, the only known quantum mixed states demonstrating entanglement reversibility are the so-called locally-flagged states \cite{Horodecki-1998b, Vollbrecht-2004a, Cornelio-2011a, Horodecki-2009a}.  
By generalizing the type of states presented in \cite{Horodecki-1998b}, we say that $\sigma_{AB}$ is an \textit{LOCC-flagged} state if there exists an LOCC instrument  $(\mc{L}_m)_m$ (i.e.~a collection of CP maps generated by an LOCC protocol \cite{Chitambar-2014b}), with $m$ enumerating the different possible public messages of the protocol, such that
(i) $\sigma=\sum_m\mc{L}_m(\sigma)$ and (ii) $\tfrac{1}{p(m)}\mc{L}_m(\sigma)=\op{\varphi_m}{\varphi_m}$ is pure, where $p(m)=\|\mc{L}_m(\sigma)\|_1$.  For such states, $E_C(\sigma)=E_D(\sigma)=\sum_m p(m) S(\tr_A\op{\varphi_m}{\varphi_m})$.  

What is the classical analog of LOCC-flagged mixed states?  Care must be taken since in the definition of key cost, Eve must be able to use her part of $p_{XYZ}$ to simulate whatever public communication Alice and Bob use to generate their parts of $p_{XYZ}$ in a formation protocol \cite{Renner-2003a}.  Given the identification of an SBI distribution as a classical pure state, we say distribution $p_{XYZ}$ is an \textit{LOPC-flagged} state if there exists an LOPC instrument $(\mc{L}_m)_m$ (i.e.~a collection of substochastic maps generated by an LOPC protocol), with $m$ enumerating the different public messages of the protocol, such that (i) $p_{XYZ}=\sum_m\mc{L}_m(p_{XYZ})$, (ii) $\tfrac{1}{p(m)}\mc{L}_m(p_{XYZ})=p(x,y|m)p(z|m)$ is SBI, where $p(m)=\|\mc{L}_m(p_{XYZ})\|_1$, and (iii) $p(z|m)p(z|m')=0$ for $m\not=m'$.  This is formally analogous to the quantum scenario except for condition (iii), which captures the ability for Eve to reproduce the public communication from her information $Z$.  Any LOPC-flagged classical state takes the form
\begin{equation}
\label{Eq:LocalFlagged}
p(x,y,z)=\sum_{M=m} p(x,y|m)p(z|m)p(m)
\end{equation}
where $M$ is generated by a public communication protocol with $I(X:Y|J_{XY|M},M)=0$ and $H(M|Z)=0$.   It immediately follows from definition that these distributions are UBI-PD, but the converse is not true.

\medskip

\noindent\textit{Conclusions.}  We have presented a class of distributions UBI-PD$\downarrow$ that are conjectured to fully characterize reversible secrecy.  Despite the complexity of these distributions, validity of this conjecture would mean that reversibility of some distribution could be decided by a single-copy analysis.  Turning back to the analogous problem of entanglement reversibility in quantum states, one might then likewise hope for a solution on the single-copy level.  Only LOCC-flagged mixed states are known to possess entanglement reversibility, and these can indeed be identified by having a particular single-copy structure.  We have proposed a classical analog to LOCC-flagged states that likewise possess reversible secrecy, but these do not constitute the full set of reversible states.  Therefore, if only LOCC-flagged quantum states possess entanglement reversibility, then the analogous statement for secrecy in classical states would not be true.  On the other hand, if entanglement and secrecy are truly on equal footing in terms of reversibility characters, then our findings might suggest the existence of reversible entanglement beyond LOCC-flagged states.

\begin{acknowledgments}
EC thanks Matthias Christandl for a helpful discussion on the intrinsic information and key cost.  EC was supported by the National Science Foundation (NSF) Early CAREER Award No. 1352326. MH is supported by an ARC Future Fellowship under Grant FT140100574.
\end{acknowledgments}

\bibliography{QuantumBib}

\onecolumngrid
\appendix

\section{Properties of the G\'{a}cs-K\"{o}rner Common Information}
In this appendix, we prove the variable $J_{XY}$ that maximizes Eq. \eqref{Eq:GK-Common-Info1} is unique up to relabeling of its range.  To do this we give an alternative characterization of $J_{XY}$, directly reminiscent of that in \cite{Gacs-1973a}.  Let $X$ and $Y$ be random variables over finite sets $\mc{X}$ and $\mc{Y}$ respectively, with joint distribution $p_{XY}$. A \textit{common partitioning of length $t$} for $XY$ are pairs of subsets $(\mc{X}_i,\mc{Y}_i)_{i=1}^t$ such that 
\begin{itemize}
\item[(i)] $\mc{X}_i\cap\mc{X}_j=\mc{Y}_i\cap \mc{Y}_j=\emptyset$ for $i\not=j$, 
\item[(ii)] $p(\mc{X}_i|\mc{Y}_j)=p(\mc{Y}_i|\mc{X}_j)=\delta_{ij}$, and  
\item[(iii)] if $(x,y)\in\mc{X}_i\times \mc{Y}_i$ for some $i$, then $p_X(x)p_Y(y)>0$. 
\end{itemize}
For a given common partitioning, we refer to the subsets $\mc{X}_i\times \mc{Y}_i$ as the ``blocks'' of the partitioning.  The subscript $i$ merely serves to label the different blocks, and for any fixed labeling, we associate a random variable $J(X,Y)$ such that $J(x,y)=i$ if $(x,y)\in\mc{X}_i\times\mc{Y}_i$.  Note that each party can determine the value of $J$ from their local information, and it is therefore called a \textit{common function} of $X$ and $Y$ \cite{Gacs-1973a}.  A \textit{maximal common partitioning} is a common partitioning of greatest length.  
\begin{proposition}
\label{Prop:Partition-Unique}
\begin{itemize}
\item[{}]
\item[(a)]  Every pair of finite random variables $XY$ has a unique maximal common partitioning.
\item[(b)]  Variable $J_{XY}$ satisfies
\[H(J_{XY})=\max_K\{H(K):0=H(K|X)=H(K|Y)\}\]
iff $J_{XY}$ is a common function for the maximal common partitioning of $XY$.
\end{itemize}

\end{proposition}
\begin{proof}
(a) Trivially $\mc{X}\times\mc{Y}$ gives a common partitioning of length one, and any common partitioning cannot have length exceeding $\min\{|\mc{X}|,|\mc{Y}|\}$; hence a maximal common partitioning exists.  To prove uniqueness, suppose that
$(\mc{X}_i,\mc{Y}_i)_{i=1}^t$ and $(\mc{X}'_i,\mc{Y}_i')_{i=1}^t$ are two maximal common partitionings.  If they are not equivalent, then there must exist some subset, say $\mc{X}_{i_0}$ such that $\mc{X}_{i_0}\subset\cup_{\lambda=1}^K\mc{X}_\lambda'$ in which $\mc{X}_{i_0}\cap\mc{X}'_{\lambda}\not=\emptyset$ for $\lambda=1,\cdots,K\geq 2$.  Choose any such $\mc{X}'_{\lambda_0}$ from this collection and define the new sets $R_{i_0}=\mc{X}_{i_0}\cap\mc{X}'_{\lambda_0}$ and $\tilde{R}_{i_0}=\mc{X}_{i_0}\setminus\mc{X}'_{\lambda_0}$, which are both nonempty since $k\geq 2$ and the $\mc{X}_\lambda$ are disjoint.  However, we also have the properties
\begin{align}
x\in\mc{X}_{i_0}&\Rightarrow p(\mc{Y}_{i_0}|x)=1;& x\in\mc{X}'_{\lambda_0}&\Rightarrow p(\mc{Y}'_{\lambda_0}|x)=1; \notag\\
x\not\in\mc{X}_{i_0} &\Rightarrow p(\mc{Y}_{i_0}|x)=0;&x\not\in\mc{X}'_{\lambda_0} &\Rightarrow p(\mc{Y}'_{\lambda_0}|x)=0.\notag
\end{align}
(Here we are implicitly using condition (iii) in the above definition by assuming that $p(x)>0$ thereby defining conditional distributions).  Therefore, $p(S_{i_0}|R_{i_0})=p(\tilde{S}_{i_0}|\tilde{R}_{i_0})=1$ and $p(S_{i_0}|\tilde{R}_{i_0})=p(\tilde{S}_{i_0}|R_{i_0})=0$, where $S_{i_0}=\mc{Y}_{i_0}\cap\mc{Y}'_{\lambda_0}$ and $\tilde{S}_{i_0}=\mc{Y}_{i_0}\setminus\mc{Y}'_{\lambda_0}$.  A similar argument shows that $p(R_{i_0}|S_{i_0})=p(\tilde{R}_{i_0}|\tilde{S}_{i_0})=1$ and $p(R_{i_0}|\tilde{S}_{i_0})=p(\tilde{R}_{i_0}|S_{i_0})=0$.  Hence, $(\mc{X}_i,\mc{Y}_i)_{i\not=i_0}^t\bigcup (S_{i_0},R_{i_0})\bigcup (\tilde{S}_{i_0},\tilde{R}_{i_0})$ is a common partitioning of length $t+1$.  But this is a contradiction since $(\mc{X}_i,\mc{Y}_i)_{i=1}^t$ is a maximal common decomposition.

(b)  Suppose that $K$ satisfies $0=H(K|X)=H(K|Y)$ so that $K=f(X)=g(Y)$ for some functions $f$ and $g$.  It is clear that $f$ and $g$ must be constant-valued for any pair of values taken from same block $\mc{X}_i\times\mc{Y}_i$ in the maximal common partitioning of $XY$.  Hence the maximum possible entropy of $K$ is then attained iff $f$ and $g$ take on a different value for each block in this partitioning.
\end{proof}

We now turn to the conditional common information $J_{XY|Z}$.  We are specifically interested in the how this information evolves under LOPC for block independent distributions.  The following provides a derviation of Eq. \eqref{Eq:CCI-Messages-Equality}.
\begin{proposition}
\label{Prop:CCI-Messages}
If $p_{XYZ}$ is BI, then so is $p_{(MX)(MY)(ZM)}$.  Moreover,
\begin{align}
I(X:Y|ZM)=I(X:Y|Z)-I(M:J_{XY|Z}|Z).
\end{align}
\end{proposition}
\begin{proof}
For a general distribution $p_{XYZ}$, it is easy to see that $H(K|Z)\leq I(X:Y|Z)$ whenever $H(K|XZ)=H(K|YZ)=0$.  Equality is obtained iff $p_{XYZ}$ is BI, and by uniqueness of the maximal conditional common function, we have that $K=J_{XY|Z}$ up to relabeling.

Now, suppose that $p_{XYZ}$ is BI and Alice locally generates message $M_1$ so that $YZ-X-M_1$.  Then
\begin{align}
\label{Eq:BI-Messages}
I(X:Y|ZM_1)&=I(M_1X:Y|Z)-I(M_1:Y|Z)\notag\\
&=I(X:Y|Z)-I(M_1:J_{XY|Z}Y|Z)\notag\\
&=I(X:Y|Z)-I(M_1:J_{XY|Z}|Z)\notag\\
&=H(J_{XY|Z}|Z)-[H(J_{XY|Z}|Z)-H(J_{XY|Z}|ZM_1)]\notag\\
&=H(J_{XY|Z}|ZM_1).
\end{align}
Since $H(J_{XY|Z}|XM_1)=H(J_{XY|Z}|YM_1)=0$, by the above discussion it follows that $p_{(XM)(YM)(ZM)}$ is BI and $J_{XY|ZM_1}$ is essentially equivalent to $J_{XY|Z}$; i.e. up to relabeling $J_{XY|Z=z}=J_{XY|Z=z,M_1=m}$ for all $m$.  The third line of Eq. \eqref{Eq:BI-Messages} gives us the desired equality in the proposition for message $M_1$.  Proceeding by induction proves the full statement for a full message $M$ generated by an arbitrarily long communication protocol. 

\end{proof}

\section{A Hierarchy of Distribution Classes}

We review the various distributions classes introduced in the paper and give different examples.  The hierarchy of the distributions is the following:
\begin{figure}[h]
\centering
\includegraphics[scale=.75]{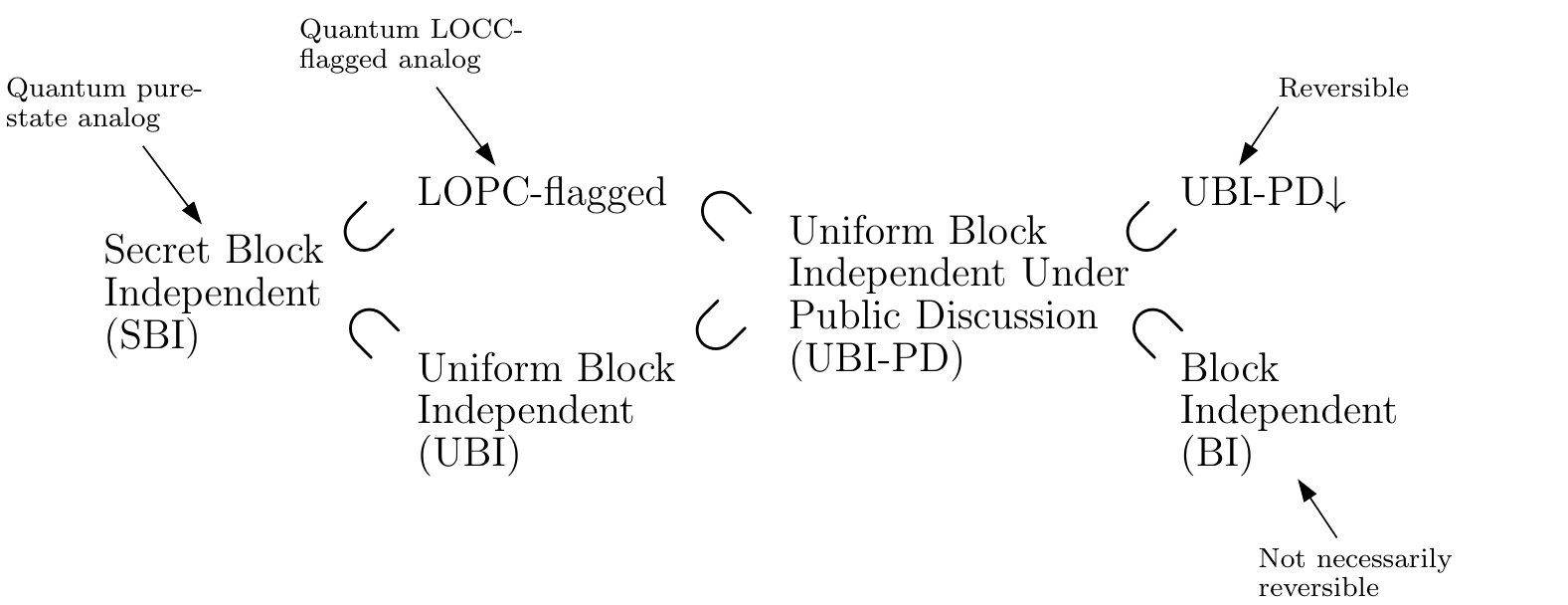}
\caption{A hierarchy of distribution classes and their relation to classes of reversible quantum states.}
\label{Fig:Dist_Hierarch}
\end{figure}

\noindent\textbf{Example Distributions:}
\begin{figure}[h]
\centering
     \includegraphics[scale=.75]{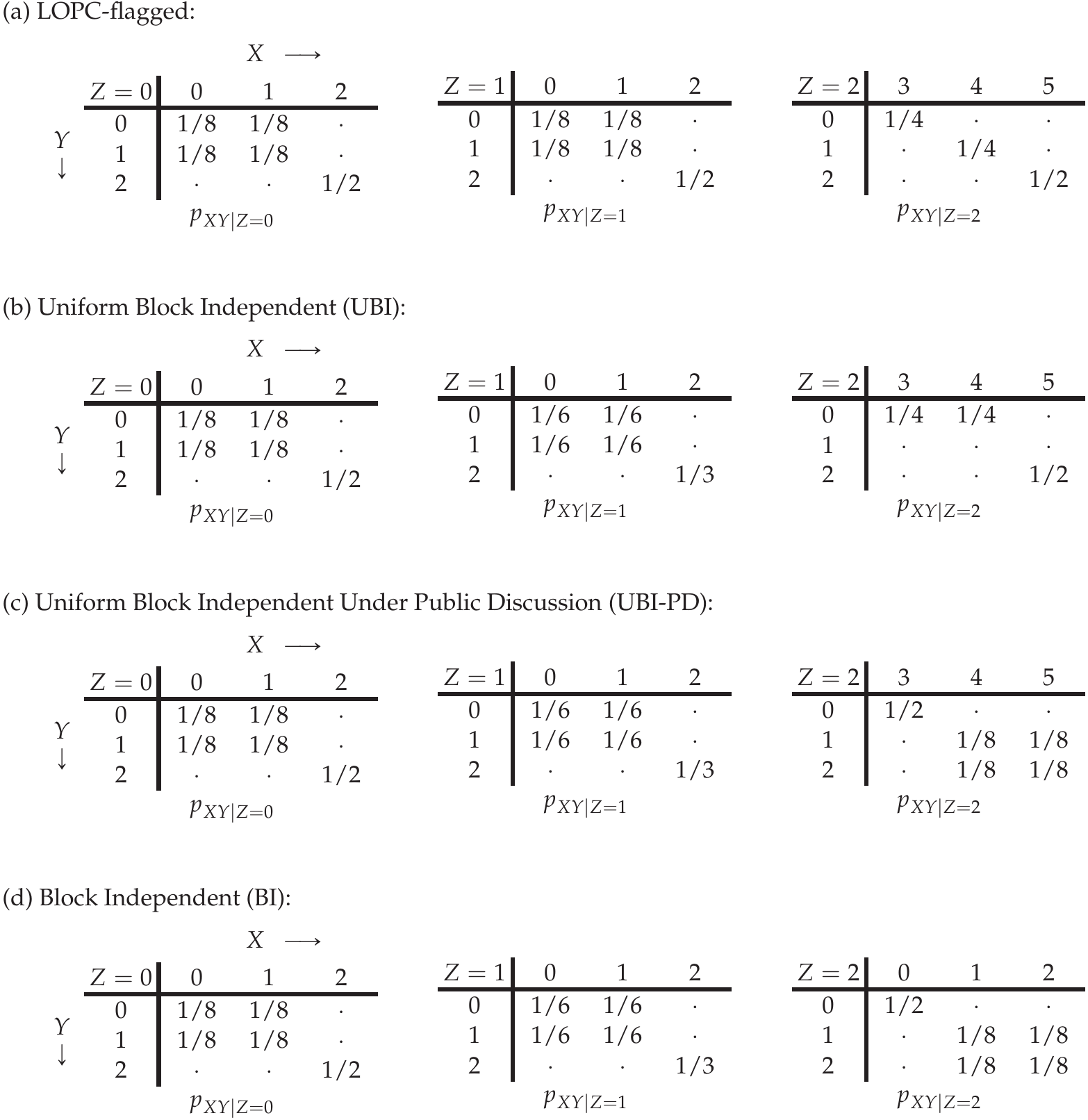}
     \caption{(a) is an LOPC-flagged distribution that is neither SBI nor UBI.  (b) is likewise a UBI distribution that is neither SBI nor LOPC-flagged.  (c) is a UBI-PD distribution that is neither UBI nor LOPC-flagged.  (d) is a BI distribution that is neither UBI-PD nor UBI-PD$\downarrow$.  Figure \ref{Fig:UBIPDdownarrow} gives a UBI-PD$\downarrow$ distribution that is not BI.}
		\label{Fig:Dist_Classes}
      \end{figure} 

\section{Conditional Double Markov Chain}

\begin{proposition}[Conditional Double Markov Chains (also Exercise 16.25 in \cite{Csiszar-2011a})]
\label{Prop:Double-Markov}
Random variables $WXYZ$ satisfy the two Markov chains $X-YZ-W$ and $Y-XZ-W$ iff $I(XY:W|J_{XY|Z}Z)=0$.
\end{proposition}

\begin{proof}
If $I(XY:W|J_{XY|Z}Z)=0$ then $I(Y:W|J_{XY|Z}Z)=0$.  The Markov chain $X-YZ-W$ follows since 
\begin{align*}
I(XY:W|J_{XY|Z}Z)&=I(X:W|YJ_{XY|Z}Z)+I(Y:W|J_{XY|Z}Z)\notag\\
&=I(X:W|YZ)+I(Y:W|J_{XY|Z}Z),
\end{align*}
where we have use the fact that $J_{XY|Z}$ is a function $X$ and $Y$ when given $Z$.  A similar argument shows that $Y-XZ-W$.

On the other hand, if the two Markov chains hold, then whenever $p_{XYZ}(x,y,z)>0$, we have
\begin{equation}
p(W=w|x,y,z)=p(w|x,z)=p(w|y,z).
\end{equation}
Hence, the conditional distribution $p(w|x,y,z)$ is constant across each block $\mc{X}_i\times\mc{Y}_i$ in the maximal common partitioning of $P_{XY|Z=z}$.  Consequently,
\[p_{W|XYZ}=p_{W|J_{XY|Z}Z},\]
and so for any $J_{XY|Z}=j$ and $Z=z$ for which $p(j,z)>0$, we have
\begin{align}
p(x,y,w|j,z)&=p(w|x,y,j,z)p(x,y|j,z)\notag\\
&=p(w|x,y,z)p(x,y|j,z)=p(w|j,z)p(x,y|j,z).
\end{align}
Thus, $I(XY:W|J_{XY|Z}Z)=0$.
\end{proof}

\section{Reversibility Conditions when $\min\{|\mc{X}|,|\mc{Y}|\}=2$}

Here we generalize Theorem \ref{Thm:Reversible-Structure-2x2}.  To do so, we will need to reference a strengthened version of Proposition \ref{Prop:DistOpt1-2x2}, which first requires some new terminology.  For a distribution $p$, let $supp[p]$ denote its support; the set of elements for which $p$ assigns a nonzero probability.  For two distributions $p_{XY}$ and $q_{XY}$ over $\mc{X}\times \mc{Y}$, we say that $q_{XY}\blacktriangleleft p_{XY}$ if, up to a permutation between $X$ and $Y$, the distributions satisfy $supp[q_X]\subset supp[p_X]$ and one of the three additional conditions: (i) $q_{XY}$ is uncorrelated, (ii) $supp[q_Y]\subset supp[p_Y]$, or (iii) $y\in supp[q_Y]\setminus supp[p_Y]$ implies that $H(X|Y=y)=0$.
\begin{lemma}[\cite{Chitambar-2014c}]
\label{Lem:DistOpt1-2xn}
Let $p_{XYZ}$ be a distribution over $\mc{X}\times\mc{Y}\times\mc{Z}$ such that $p_{XY|Z=z_1}\blacktriangleleft p_{XY|Z=z_0}$ for some $z_0,z_1\in\mc{Z}$.  If there exists some pair $(x,y)\in supp[p_{X|Z=0}]\times supp[p_{Y|Z=0}]$ for which $p(x,y|z_1)>0$ but $p(x,y|z_0)=0$, then $K_D(p_{XYZ})<I(X:Y|Z)$.  
\end{lemma}

\noindent Using this lemma, we are able to provide a full solution to the secrecy reversibility problem when one of the parties has a binary random variable.

\begin{theorem}
\label{Thm:Reversible-Structure-2xn}
Suppose that $\min\{|\mc{X}|,|\mc{Y}|\}=2$.  Then $p_{XYZ}$ satisfies $K_D(p_{XYZ})=K_C(p_{XYZ})$ iff $p_{XYZ}$ is UBI-PD.
\end{theorem}
\begin{proof}
It suffices to prove necessity.  If $K_C(p_{XYZ})=K_D(p_{XYZ})$ then by Lemma \ref{Lem:Structure1}, $p_{XY\overline{Z}}$ must be block independent, where $I(X:Y\downarrow Z)=I(X:Y|\overline{Z})$.  By the same reasoning as in Theorem \ref{Thm:Reversible-Structure-2x2}, $K_C(p_{XYZ})=K_D(p_{XYZ})$ implies that $K_D(p_{XY\overline{Z}})= I(X:Y|\overline{Z})$.  Hence we will apply Lemma \ref{Lem:DistOpt1-2xn} on the equality $K_D(p_{XY\overline{Z}})= I(X:Y|\overline{Z})$ to derive a necessary condition for $p_{XY\overline{Z}}$.

Without loss of generality, assume that $|\mc{X}|=2$.  Since $p_{XY\overline{Z}}$ is BI, every conditional distribution $p_{XY|\overline{Z}=z}$ is either:
\begin{itemize}
\item[(I)] Uncorrelated $I(X:Y|Z=z)=0$ or 
\item[(II)] Correlated and satisfying $H(X|Y,\overline{Z}=z)=0$.
\end{itemize}
Suppose now that $H(X|Y=y)>0$ for some $y\in\mc{Y}$, but nevertheless $y$ is a possible event in correlated distribution $p_{XY|\overline{Z}=z}$.  The latter means that $p(x,y|z)=0$ for some $x\in\{0,1\}$.  However, $H(X|Y=y)>0$ implies the existence of some $\tilde{z}\not=z$ such that $p(\overline{x},y|\tilde{z})>0$, where $\overline{x}=x\oplus 1$.  With $p_{XY|\overline{Z}=z}$ having correlations, then $supp[p_{X|\overline{Z}=\tilde{z}}]\subset supp[p_{X|\overline{Z}=z}]$, and since $p_{XY|\overline{Z}=\tilde{z}}$ has either form (I) or (II), it follows that $p_{XY|\overline{Z}=\tilde{z}}\blacktriangleleft p_{XY|\overline{Z}=z}$.  But then Lemma \ref{Lem:DistOpt1-2xn} implies that $K_D(p_{XY\overline{Z}})< I(X:Y|\overline{Z})$, which contradicts our assumption.  Therefore, if $y$ is a possible event in any correlated conditional distribution $p_{XY|\overline{Z}=z}$, then $H(X|Y=y)=0$.  Consequently, we can define the following message for Bob and maximal conditional common functions:
\begin{align}
&M(y)=\begin{cases}0\quad\text{if $H(X|Y=y)>0$}\\1\quad\text{if $H(X|Y=y)=0$}\end{cases}&J_{XY|\overline{Z}}(x,y,z)&=\begin{cases}0\quad\text{if $p_{XY|\overline{Z}=z}$ is uncorrelated}\\x\quad\text{$p_{XY|\overline{Z}=z}$ is correlated}\end{cases}\notag\\
&J_{XY|\overline{Z}M}(x,y,z,m)=\begin{cases}0\quad\text{if $m=0$}\\x\quad\text{if $m=1$}.\end{cases}
\end{align}
It is obvious that $I(J_{XY|Z}:M|Z)=0$ since $J_{XY|Z}=0$ for all $z$ whenever $p_{XY|Z=z}$ is uncorrelated, and $M=1$ for all $z$ whenever $p_{XY|Z=z}$ is correlated.  Also, $J_{XY|ZM}$ becomes a shared variable for Alice and Bob since it can be computed both by Alice and Bob given $M$.  We thus, see that $p_{XY\overline{Z}}$ is UBI-PD.

\end{proof}

\section{Calculation of Theorem \ref{Thm:Gap2x2}}

First recall that for a two-qubit state $\rho$, its concurrence is defined by $C(\rho)=\sqrt{\lambda_1}-\sqrt{\lambda_2}-\sqrt{\lambda_3}-\sqrt{\lambda_4}$, where the $\lambda_i$ are the non-increasing eigenvalues of the operator $\rho\tilde{\rho}$, with $\tilde{\rho}=(\sigma_2\otimes\sigma_2)\rho^*(\sigma_2\otimes\sigma_2)$ \cite{Wootters-1998a}.  Here, $\rho^*$ is the complex conjugate of $\rho$ in the computational basis, and the $\sigma_i$ are the Pauli matrices: $\sigma_1=\left(\begin{smallmatrix}0&1\\1&0\end{smallmatrix}\right)$, $\sigma_2=\left(\begin{smallmatrix}0&-i\\i&0\end{smallmatrix}\right)$, and $\sigma_3=\left(\begin{smallmatrix}1&0\\0&-1\end{smallmatrix}\right)$.  For a two-qubit state with concurrence $C(\rho)$, its entanglement of formation is given by $\mathsf{E}(C(\rho))$, where $\mathsf{E}(x):=h(\tfrac{1}{2}[1-\sqrt{1-x^2}])$ and $h(x):=-x\log x-(1-x)\log(1-x)$.  Note that $\mathsf{E}(x)$ is strictly convex in $x$.

When $|\mc{X}|=|\mc{Y}|=2$, reversible distributions are UBI-PB$\downarrow$.  If $K_D(p_{XYZ})>0$, then the distribution is UBI-PB, and we wish to show:
\begin{align}
\label{Eq:Appendix-2x2eqn}
K_D(p_{XYZ})&=\sum_{z\in\mc{Z}}p(z)\mathsf{E}\left(2\sqrt{p(0|z)p(1|z)}\right)\notag\\
E_F(\rho_{XY})&=\mathsf{E}\left(2\sum_{z\in\mc{Z}}p(z)\sqrt{p(0|z)p(1|z)}\right).
\end{align}
Up to a relabeling of $x$, a general UBI-PB distribution in $2\times 2$ is given by $p(x,y|z)=\delta_{xy}p(x|z)$ for $x,y\in\{0,1\}$ and arbitrary $p(x|z)$.  Then 
$\rho_{AB}=\sum_z p(z)\op{\varphi_z}{\varphi_z},$
where $\ket{\varphi_z}=\sum_{x=0}^1\sqrt{p(x|z)}\ket{xx}$.  This corresponds to a single qubit density matrix 
\begin{align}
\label{Eq:omega}
\omega&=\begin{pmatrix}\sum_zp(z)p(0|z)&\sum_zp(z)\sqrt{p(0|z)p(1|z)}\\\sum_zp(z)\sqrt{p(0|z)p(1|z)}&\sum_zp(z)p(1|z)\end{pmatrix}\notag\\
&=\begin{pmatrix}\sum_zp(z)p(0|z)&\tfrac{1}{2}\sum_zp(z)\sqrt{C(\varphi_z)}\\\tfrac{1}{2}\sum_zp(z)\sqrt{C(\varphi_z)}&\sum_zp(z)p(1|z)\end{pmatrix}.
\end{align}
It can be seen that $\sigma_2\omega^*\sigma_2=\sigma_1\omega\sigma_1$.  Hence, the concurrence of $\rho_{AB}$ can be computed from the eigenvalues of the $2\times 2$ matrix $\omega\tilde{\omega}=\omega\sigma_1\omega\sigma_1$, which are
\[\left(\sqrt{p(0)p(1)}\pm \sum_zp(z)\sqrt{p(0|z)p(1|z)}\right)^{2}.\]
The Cauchy-Schwarz Inequality then gives that 
\begin{equation}
C(\rho_{AB})=\sqrt{\lambda_{\max}}-\sqrt{\lambda_{\min}}=2\sum_zp(z)\sqrt{p(0|z)p(1|z)}=\sum_zp(z)C(\varphi_z).
\end{equation}
Since $K_D(p_{XYZ})=H(J_{XY|Z}|Z)=\sum_zp(z)E(C(\varphi_z))$, the calculation of Eq. \eqref{Eq:Appendix-2x2eqn}.  Note that by strict convexity of $\mathsf{E}(x)$, we have 
\[\sum_{z\in\mc{Z}}p(z)\mathsf{E}\left(2\sqrt{p(0|z)p(1|z)}\right)=\mathsf{E}\left(2\sum_{z\in\mc{Z}}p(z)\sqrt{p(0|z)p(1|z)}\right)\]
iff $p(0|z)p(1|z)$ is constant for all $z$.  This implies that $H(X|Z=z)$ is constant for all $z\in\mc{Z}$.

\end{document}